\documentclass[preprint,12pt,authoryear,nopreprintline]{elsarticle}

\usepackage[colorlinks,citecolor=blue,urlcolor=blue]{hyperref}
\usepackage{amsthm,amsmath,amsfonts,amssymb,bbm}
\usepackage{accents}
\usepackage{algpseudocode}
\usepackage{graphicx}
\usepackage{mathtools}
\usepackage{subcaption}
\usepackage{tikz}

\DeclareMathOperator*{\argmin}{\arg\!\min}

\usepackage{algorithm}
\usepackage{algpseudocode}

\newcommand{\norm}[1]{\left\lVert#1\right\rVert}

\newcommand{\expec}[1]{\mathbb{E}\left[#1\right]}

\newcommand{\Var}[1]{\textrm{Var}\left[#1\right]}

\newcommand{\hajek}[1]{\accentset{\circ}{#1}}
\renewcommand{\P}{\mathbb{P}}
\renewcommand{\d}{\text{d}}

\newcommand{\R}{\mathbb{R}}

\newcommand{\N}{\mathbb{N}}

\newcommand{\abs}[1]{\left|#1\right|}

\newcommand{\eset}[1]{{\left\{#1\right\}}}

\renewcommand{\t}[1]{{\textrm{#1}}}
\renewcommand{\c}[1]{{\mathcal{#1}}}

\newcommand{\Prob}[1]{\P\left[#1\right]}

\newcommand{\diamO}{\t{diam}\,\Omega}

\newtheoremstyle{def}
  {12pt}   
  {6pt}   
  {\normalfont}  
  {0pt}       
  {\bfseries} 
  {.}         
  {5pt plus 1pt minus 1pt} 
  {}          

\theoremstyle{def}
\newtheorem{theorem}{Theorem}
\newtheorem{lemma}{Lemma}
\newtheorem{proposition}{Proposition}
\newtheorem{assumption}{Assumption}

\graphicspath{{figures/}}

\def\centerarc[#1](#2)(#3:#4:#5)
    { \draw[#1] ($(#2)+({#5*cos(#3)},{#5*sin(#3)})$) arc (#3:#4:#5); }

\begin{document}

\begin{frontmatter}

\title{Medoid splits for efficient random forests in metric spaces\tnoteref{title}}
\author{Matthieu Bulté\fnref{label1}}
\ead{mb@math.ku.dk}
\author{Helle Sørensen\fnref{label1}}
\ead{helle@math.ku.dk}
\affiliation[label1]{organization={Department of Mathematical Sciences, University of Copenhagen}}

\begin{abstract}
    This paper revisits an adaptation of the random forest algorithm for Fréchet regression, addressing the challenge of regression in the context of random objects in metric spaces. Recognizing the limitations of previous approaches, we introduce a new splitting rule that circumvents the computationally expensive operation of Fréchet means by substituting with a medoid-based approach. We validate this approach by demonstrating its asymptotic equivalence to Fréchet mean-based procedures and establish the consistency of the associated regression estimator. The paper provides a sound theoretical framework and a more efficient computational approach to Fréchet regression, broadening its application to non-standard data types and complex use cases. 
\end{abstract}


\end{frontmatter}

\section{Introduction}\label{sec-intro}

We study the extension of random forest to regression situations where the response takes values in a metric space. The usual expectation is replaced by the Fr\'echet mean, which can be be very computationally intensive with restrictions on possible applications to follow. Emphasis is therefore on a new medoid-based splitting rule used in the individual trees. Application of this new splitting rule speeds up computations without compromising consistency or finite-sample prediction quality and makes it possible to analyzer larger datasets in terms of sample size and dimension.  

Random forest is a statistical learning method introduced by \citet{breiman_random_2001} for classification and regression. It constructs predictions by averaging the predictions from a randomized ensemble of regression trees fitted on a training dataset. Random forests have proven to perform well in practice on complex and high-dimensional datasets with little parameter tuning, making them an ubiquitous method in the machine learning toolset.

With the success of statistical analysis, more domains are being analyzed, leading to the study of more complex use cases and non-standard data types. A first generalization that has emerged from this need is the study of random variables taking values in manifolds. More recently, a further generalization has started to be studied, in which the random variables, referred to as random objects, take values in a metric space. This general setting does not require anything more than a distance function, and does not require common algebraic structures such as those found in vector spaces or on manifolds. A core quantity in the study of random objects is the Fr\'echet mean, as defined by \citet{frechet_les_1948}, which generalizes the concept of the expected value to metric spaces by defining it as the minimizer of the expected squared loss. This notion allows to generalize the definition of the variance and further construct tools to study the variability of random objects.

The lack of algebraic structure in the study of random objects makes the development of regression algorithms particularly challenging. Efforts in adapting regression algorithms to non-Euclidean objects have mainly been done for the case where the response takes values in finite-dimensional differentiable Riemannian manifolds. This has led to the development and analysis of parametric models exploiting the local structure of the space, such as geodesic regression approaches \citep{thomas_fletcher_geodesic_2013}, but also fully non-parametric approaches via adaptations of the Nadaraya-Watson estimator \citep{davis_population_2010,hinkle_polynomial_2012,pelletier_non-parametric_2006, yuan_local_2012,hein_robust_2009}. More recently, \citet{petersen_frechet_2019} proposed a general framework for regression in metric spaces with Euclidean covariates in which the regression method specifies a weighting scheme. For a given input, the training samples are reweighted following that scheme and the prediction is formed by computing the weighted Fréchet mean. They then define a global weighting scheme generalizing the classical linear regression and a local scheme which generalizes local linear regression. This original work paved the way to the construction and study of more flexible regression techniques adapted to Fréchet regression.

Random forest is a well-suited candidate for adaptation to Fréchet regression. Several authors have already used random forests outside the classical setup of Euclidean response, for example for function spaces \citep{nerini_classifying_2007,fu_functional_2021} by either working directly in the Hilbert space $L_2$ or via a transformation to $L_2$, but these methods lack a theoretical analysis and their adaptation to metric space data is unclear. Another attempt was made by \citet{capitaine_frechet_2020} by direct translation of the random forest algorithm via a two-stage approach in which the forest prediction is constructed as the Fréchet mean of the prediction of the trees. Their approach has been successfully applied to general settings in which both the covariates and response are allowed to be in a metric space and proposes a novel splitting rule that reduces the computational cost of the tree-fitting procedure. Unfortunately, this two-stage procedure complicates theoretical analysis, and to our knowledge, no consistency result for their method has been developed. 

Another route for adapting random forest to Fréchet regression is to fit a random forest and aggregate the trees into weights for the observations, which can be used to construct a prediction. This approach was developed by \citet{meinshausen_quantile_2006} for quantile regression, and \citet{athey_generalized_2019} for conditional Z-estimation under the name of Generalized Random Forest (GRF).  It has been exploited by \citet{qiu_random_2022} to construct a random forest adapted to Fréchet regression. The authors combined this with the splitting rule proposed by \citet{capitaine_frechet_2020} to alleviate some of the computational burden of tree fitting. The paper presents results about asymptotic properties of the estimator, but proofs are not provided. Moreover, the results regarding rates of convergence and a central limit theorem depend on strong assumptions regarding the random forest algorithm itself.

One essential ingredient in the adaptation of a random forest algorithm is the construction of a suitable splitting rule used to recursively partition the covariate space in the fitting of individual trees. In their work, \citet{athey_generalized_2019} propose a general splitting rule for Z-estimation relying on gradients of the objective function to build splits capturing the heterogeneity of the response. This approach cannot be used for Fréchet regression since it requires a differentiability structure which is not available. A possible approach in Fréchet regression is to use the standard splitting rule when the response is in $\R$, namely by finding the splitting position that results in the highest reduction in variance in the subgroups. However, for each split, this requires computing a large number of Fréchet variances and hence Fréchet means. In many cases, computing the Fréchet mean is expensive, making this approach unusable. The approach proposed by \citet{capitaine_frechet_2020} and used in \citet{qiu_random_2022} based on 2-means clustering for efficient splitting reduces the number of Fréchet means required to a great extent but, as highlighted by the authors, is still computationally intensive. 

In this article, we revisit the algorithm proposed by \citet{qiu_random_2022}, which implements GRF for Fréchet regression. Our main contribution is a new splitting rule which replaces the computation of Fréchet means with medoids, i.e., sample members. This, after the computation of a distances matrix, allows to find splits as efficiently as in the Euclidean case. We denote the adapted procedure Metric Random Forest (MRF). We prove that the new splitting rule is asymptotically equivalent to using the Fréchet mean-based procedure. Furthermore, we prove the consistency of the MRF based on classical M-estimation results from \citet{vaart_asymptotic_1998} and following the proof strategy presented in \citet{wager_estimation_2018}, and illustrate the benefit of the new splitting rule with numerical experiments. 

The paper is organized as follows: Section~\ref{sec_background} includes background material on Fréchet regression and existing work on random forest in metric spaces. The new medoid-based splitting rule is introduced in Section~\ref{sec_splitting_criterion}, and the consistency result for the MRF estimator is stated in Section~\ref{sec-asymptotics}. We present simulations from three scenarios in Section~\ref{sec_numerical} and conclude in Section~\ref{sec-conclusion}. Proofs are presented in the appendix.  

\section{Background}\label{sec_background}

\subsection{Fréchet regression}\label{subsec_frechet_regression}

Let $(\Omega, d)$ be a compact metric space. Consider a pair of random variables $(X, Y)$, where $X$ is a vector of predictors taking value in $[0, 1]^d$ with density bounded away from $0$ and $\infty$, and $Y$ is a response taking value in $\Omega$. Following the idea of \citet{frechet_les_1948}, the notions of mean and variance can be generalized to random objects to obtain the Fréchet mean and Fréchet variance,
\begin{equation}\label{eq-fr}
  \omega_\oplus = \argmin_{\omega \in \Omega} \expec{d(Y, \omega)^2}, \qquad V_\oplus = \expec{d(Y, \omega_\oplus)^2}.
\end{equation}
Building upon these generalized concepts of mean and variance, \citet{muller_peter_2016} define the Fréchet regression function, denoted $m$, in terms of the conditional distribution of $Y$ given $X=x$,
\begin{equation} \label{eq-cond-fr}
  m(x) = \argmin_{\omega \in \Omega} M(\omega; x), \qquad M(\omega; x) = \expec{d(Y, \omega)^2 \mid X=x},
\end{equation}
where $M$ is called the conditional Fréchet function. 

We are interested in fitting a regression model with response $Y$ and predictors $X$. Given a dataset $\c{D} = \eset{(X_i, Y_i)}_{i=1}^n$ of independent random pairs following the same distribution as the prototypical pair $(X, Y)$ and a point $x \in [0,1]^d$ for which a prediction is to be obtained, the goal is to construct an estimator $m_n(x)$ of the Fréchet regression function $m(x)$. One generally applicable solution to this problem is to estimate the conditional Fréchet function $M(\cdot; x)$ with a data-dependent objective function $M_n(\cdot; x)$ converging to $M(\cdot; x)$ in such a way that the maximizer $m_n(x)$ of $M_n(\cdot; x)$ also converges to $m(x)$. A natural approach consists in using a re-weighting strategy to approximate the conditional expectation in (\ref{eq-cond-fr}), yielding an estimator of the form
\begin{equation}
  M_n(\omega; x) = \sum_{i=1}^n w_i(x) d(\omega, Y_i)^2,
\end{equation}
where the weights $\eset{w_i(x)}_{i=1}^n$ depend on the data and the point $x$ for which the prediction is being made. The construction of such weights has been the subject of a large corpus of literature. A common and well studied method is the Nadaraya-Watson estimator weights which defines weights through a kernel function $K$ and a bandwidth parameter $h \in (0, \infty)$, giving $w_i(x) = K_h(X_i - x)$ with $K_h(\cdot) = h^{-1}K(\cdot / h)$. This method is well understood and easy to implement, but it requires a crucial tuning of the bandwidth parameter $h$, suffers from the curse of dimensionality and is not adaptive, in the sense that the weights are invariant to the part of the covariate space where the prediction is made. This lack of adaptivity makes the Nadaraya-Watson estimator practically ill-suited in case of heteroskedasticity of the response for different values of the predictors. 

\subsection{Random forest in metric spaces}\label{subsec_random_forest}

The original and commonly used random forest algorithm for regression proposed by \citet{breiman_random_2001} relies on constructing an ensemble of randomized regression trees. A prediction at a given point is then formed by averaging the predictions across the ensemble of trees. The fitting of the trees is randomized by bootstrap aggregation, meaning that each regression tree is fitted on a random subset of size $s$ of the training data. Each tree is built greedily from a recursion of splits, each attempting to partition the training sample available at the node into two subsets of minimal prediction error \citep{breiman_classification_1984}. The procedure results in a partitioning of the space of predictors into a collection of rectangles called \textit{leaves}. The resulting random forest estimator then takes the form
\begin{equation}
  \t{RF}(x) = \frac{1}{B} \sum_{b=1}^B \t{T}_b(x) \qquad\t{with} \qquad \t{T}_b(x) = \sum_{i=1}^n \frac{\mathbbm{1}\eset{X_i \in L_b(x)}}{ \#\eset{j : X_j \in L_b(x) } } Y_i,
\end{equation}
where $B$ is the size of the ensemble, $\t{T}_b$ is the $b$-th regression tree and $L_b(x)$ is the subset of training points used for fitting $\t{T}_b$ falling in the same leaf as $x$.

In order to generalize this idea to the context of Fréchet regression, \citet{capitaine_frechet_2020} directly translate the average-of-averages construction of the random forest by replacing averages with Fréchet means. The resulting estimator, called the Fréchet Random Forest (FRF), is constructed as the Fréchet mean of an ensemble of randomized Fréchet trees (FTs) in which each tree is built to minimize the prediction error, as done in the Euclidean regression setup.

In this manuscript, a different approach, first proposed in \citet{qiu_random_2022}, is followed to constructing a random forest prediction. It does not aggregate the prediction of each individual tree; instead the results of the fitted ensemble are aggregated following the idea presented in \citet{meinshausen_quantile_2006} and \citet{athey_generalized_2019}. For a given $x$, each tree contributes with a weighting of the observations reflecting the partitioning of the predictor space fitted by the tree. 
Consider a source of auxiliary randomness $\xi$ and a subset of the data $D \subset \c{D}$, we denote by $\eset{w_i(\cdot; \xi, D)}_{i=1}^n$ the set of weights resulting from fitting a single tree on $D$. These weights take the form
\begin{equation}\label{eq-tree-weight}
  w_{i}(x; \xi, D) = \frac{\mathbbm{1}\eset{X_i \in L(x; \xi, D)}}{ N(L(x; \xi, D))},
\end{equation}
where $N(C) = \# \eset{j : X_j \in C }$ for any $C \subset [0,1]^d$, and $L(x; \xi, D)$ represents the observations in $D$ falling in the same leaf as $x$. Note that for observations $(X_i, Y_i) \in \c{D} \backslash D$, the associated weight function $w_{i}(\cdot; \xi, D)$ is zero. Given $B$ random subsets $\eset{D_b}_{b=1}^B$ of $\c{D}$ and auxiliary randomness $\eset{\xi_b }_{b=1}^B$, the weights constructed by the individual trees are averaged within the ensemble to construct an adaptive set of weights $\eset{w_i}_{i=1}^n$ for a given prediction point $x$ given by
\begin{equation}
  w_i(x) = \frac{1}{B} \sum_{b=1}^B w_{i}(x; \xi_b, D_b).
\end{equation}
These weights are then used to form an estimator of the conditional Fréchet function defined in (\ref{eq-cond-fr}),
\begin{equation}\label{eq-mrf-ff}
  M_n(\omega; x) = \sum_{i=1}^n w_i(x) d(\omega, Y_i)^2.
\end{equation}
Minimizing this adaptively weighted sum yields our Metric Random Forest (MRF) estimator of the Fréchet regression function in (\ref{eq-cond-fr}),
\begin{equation} \label{eq-mrf}
  m_n(x) = \argmin_{\omega \in \Omega} M_n(\omega; x).
\end{equation}
It is important to note that, in general, this optimization problem needs not have a unique solution. In practice, a user of the MRF must provide a solver for (\ref{eq-mrf}) that can return a unique solution based on any preferred heuristic implemented by the user (e.g., minimal norm, first in lexicographic order). Since this technicality does not influence the implementation or analysis of the MRF algorithm, we can safely assume that $m_n(x)$ exists and is unique. Note that we later assume in Assumption \ref{ass-ex-pop} that the population version of this quantity exists and is unique.

\section{Splitting criterion in metric spaces}\label{sec_splitting_criterion}


The recursive partitioning used to construct each individual tree is driven by the choice of a splitting criterion. Considering an arbitrary cell $C \subset [0,1]^d$, the goal at each partitioning step is to find the feature index $j \in \eset{1, \ldots, d}$ and cut position $z \in [0,1]$ defining an axis-aligned plane along which to split the cell into two subcells, $C_l$ and $C_r$, given by
\begin{equation*}
  C_l = \eset{ x \in A : x^{(j)} \leq z } \qquad\t{and}\qquad C_r = \eset{x \in A : x^{(j)} > z},
\end{equation*}
where $x^{(j)}$ is the $j$th coordinate of $x$. Intuitively, each splitting pair $(j, z)$ should be chosen as to improve the estimation of the Fréchet regression function $m$. This is done by minimizing the CART splitting criterion of \citet{breiman_classification_1984}. Let us denote by $\hat\omega_\oplus(C)$ and $\hat V_\oplus(C)$ the empirical Fréchet mean and variance computed on the observations with $X$ contained in $C$,
\begin{align*}
  \hat \omega_\oplus(C) &= \argmin_{\omega \in \Omega} \sum_{i=1}^n \mathbbm{1}\eset{ X_i \in C }d(Y_i, \omega)^2 ,\\
  \hat V_\oplus(C) &= \sum_{i=1}^n \frac{\mathbbm{1}\eset{ X_i \in C }}{N(C)}d(Y_i, \hat \omega_\oplus(C))^2.
\end{align*}
The CART splitting criterion is then given by
\begin{equation}\label{eq-cart}
  \t{err}(j, z) 
  = 
  \frac{N(C_l)}{N(C)} \hat V_\oplus(C_l)
  +
  \frac{N(C_r)}{N(C)} \hat V_\oplus(C_r).
\end{equation}
Note that minimizing this criterion corresponds to finding a split that maximizes the decrease in the weighted average of Fréchet variance accross the splitted subcells. A splitting pair $(j^\star, z^\star)$ can be found by iterating over the feature index $j \in \{1, \ldots, d\}$ and possible threshold values $z \in \{ X^{(j)}_i : X_i \in C \}$ determined from values of coordinate $j$ taken by the samples present in the cell $C$.

As highlighted by \citet{athey_generalized_2019} and \citet{capitaine_frechet_2020}, this greedy search algorithm may be computationally prohibitively expensive since it requires for each split and candidate threshold to compute Fréchet means $\hat m_l$ and $\hat m_r$, which overall requires $O(d n^2)$ computations of Fréchet means. In many cases, computing a Fréchet mean can only be done approximately through computationally expensive algorithms. This is the case for computation of Fréchet means in Wasserstein spaces \citep{panaretos_invitation_2020}, in function spaces equipped with the amplitude and phase variation distances \citep{srivastava_functional_2016} or in Riemannian manifolds.

To overcome this computational limitation, \citet{athey_generalized_2019} introduce the \textit{gradient tree algorithm}. This algorithm utilizes an alternative splitting criterion based on a linearization of the target optimization problem in regression and enables faster calculation of the relevant quantities in each candidate subcell. However, this approach cannot be applied to general metric spaces where no vector space structure is available for differentiability. In \citet{capitaine_frechet_2020}, an alternative splitting criterion is proposed which offers computational advantages over the CART splitting criterion and is also usable in metric spaces. This approach, which is also used in \citet{qiu_random_2022}, reduces the number of Fréchet mean computations by testing a single possible partitioning for each coordinate, where the partitioning is found by performing a $2$-means clustering of the observations in the cell based on the tested coordinate. This allows to compute only $2d$ Fréchet means per split. However, this still results in $O(dn)$ computations of Fréchet means which renders the algorithm unusable in situations where the Fréchet mean is expensive to compute. Notice that trees are fitted to subsamples of $\mathcal D$ of size $s$, but the abovementioned computational complexities are still valid with $n$ replaced by the size of the subsamples.

We propose a solution that is usable in any metric space and completely avoids the need to compute Fréchet means during the fitting of the tree. To do that, we replace the minimization of the Fréchet function over the entire metric space by a minimization over the available sample. Denoting by $\Omega_n$ the set of observed responses, $\Omega_n = \eset{Y_1, \ldots, Y_n}$, we define an approximation to the empirical Fréchet mean $\hat \omega_\oplus(C)$, the \textit{Fréchet medoid estimator} $\tilde \omega_\oplus(C) \in \Omega_n$, by
\begin{equation*}
  \tilde \omega_\oplus(C) = \argmin_{\omega \in \Omega_n} \sum_{i=1}^n d(\omega, Y_i)^2 \mathbbm{1}\eset{X_i \in C}.
\end{equation*}
Computation of distances between all pairs of elements in the training data is requires, but once these distances have been computed, the Fréchet medoid is straightforward to find. Using the Fréchet medoid hence provides an efficient alternative to the empirical Fréchet mean in the computation of the CART splitting criterion (\ref{eq-cart}).

The algorithm we propose to fit the random forest is based on replacing the computation of the empirical Fréchet mean by the Fréchet medoid in the split-finding procedure. As a first step, we pre-compute the pairwise distances matrix $\Delta \in \R^{n \times n}$ with $\Delta_{ij} = d(Y_i, Y_j)$. Then, at each step, we find the split $(j, z)$ minimizing the approximate splitting criterion
\begin{equation*}
  \tilde{\t{err}}(j, z) 
  = \frac{1}{N_n(C)}\sum_{i=1}^n \Delta_{i, i^\star(C_l)}^2 \mathbbm{1}\eset{ X_i \in C_l }
   + \Delta_{i, i^\star(C_r)}^2 \mathbbm{1}\eset{ X_i \in C_r },
\end{equation*} 
where $i^\star$ is the index of the Fréchet medoid, $\tilde\omega_\oplus(C) = Y_{i^\star(C)}$.

Provided that the discrete set $\Omega_n$ grows dense in $\Omega$ and further technical conditions, the Fréchet medoid $\tilde\omega_\oplus(C)$ consistently estimates the true Fréchet mean (Proposition~\ref{thm-consistency-isfm}), and furthermore the associated splitting criterion $\tilde{\t{err}}(j, z)$ inherits this consistency. Notice that Assumption 1 is stated in Section~\ref{sec-asymptotics}, and that a proof for the proposition is given in the appendix. 


\begin{proposition} \label{thm-consistency-isfm}
  Let $Y_1, \ldots, Y_n$ be independent copies of a prototypical random  variable $Y$ taking values in a compact metric space $(\Omega, d)$ with Fréchet mean $\omega_\oplus$. Assume that the Fréchet mean is well-separated as defined in Assumption \ref{ass-ex-pop}, and that the Fréchet mean is a \textit{possible value} for $Y$, in the sense that for every $\varepsilon > 0$,
  \begin{equation*}
      \Prob{ d(Y, \omega_\oplus) < \varepsilon} > 0.
  \end{equation*} 
  Then, the Fréchet medoid estimator is a consistent estimator of the Fréchet mean, that is, $d(\tilde\omega_\oplus, \omega_\oplus) = o_P(1)$.
\end{proposition}

Together with the Lipschitz continuity of the squared distance function (Lemma~\ref{lem-lipschitz} in the appendix), one can see that for some constant $K > 0$ independent of $n$,
\begin{equation*}
  \abs{\tilde{\t{err}}(j, z) - \t{err}(j, z) }
  \leq K \eset{d(\hat \omega_\oplus(C_l), \tilde\omega_\oplus(C_l)) + d(\hat \omega_\oplus(C_r), \tilde\omega_\oplus(C_r)) }.
\end{equation*}
Since $\tilde\omega_\oplus$ and $\hat\omega_\oplus$ are both consistent estimators of $\omega_\oplus$, this implies that the proposed splitting criterion converges to the CART splitting criterion.

\section{Asymptotic Theory} \label{sec-asymptotics}
We now consider the pointwise consistency of the MRF. We show that for a fixed vector of covariates $x \in [0,1]^d$, the MRF (\ref{eq-mrf}) converges in probability to  the output of the regression function (\ref{eq-cond-fr}). Consistency of the random forest estimators have been established in various settings: for simplified forest models \cite{breiman_consistency_2004,biau_performance_2008,scornet_asymptotics_2016}, forest with real number responses \citep{wager_estimation_2018,scornet_asymptotics_2016} and also in the more general Z-estimator setting \citep{athey_generalized_2019}. However, due to the reliance on the outcome taking values in an Euclidean space, both in the assumptions and in the proofs of their theoretical results, existing asymptotic analysis, in particular that of the GRF, cannot be directly applied to the MRF setting.

Instead, we develop a proof combining results from standard M-estimation theory and existing results of pointwise consistency of the classical Euclidean random forest algorithm.
Based on standard results in M-estimation theory (see \citet{vaart_asymptotic_1998}), the first step to establishing pointwise consistency of the MRF is showing the convergence of the approximate objective function $M_n$ given in (\ref{eq-mrf-ff}) to the conditional Fréchet function $M$ from (\ref{eq-cond-fr}) uniformly in $\omega \in \Omega$. We show this based on results from \citet{newey_uniform_1991} and the pointwise convergence of the MRF objective function based on the proof presented in \citet{wager_estimation_2018}. Their proof combines the Lipschitz continuity of the objective function with an assumption on the construction of the trees which implies that the diameter of the leaves of the tree shrinks as more datapoints become available (Specification 1 from \citet{athey_generalized_2019}, see below).

Similarly to \citet{wager_estimation_2018}, we study in our analysis a theoretical construction of the random forest in which the bagging procedure is replaced by fitting trees on each possible subsets $D$ of size $s$ of the dataset $\c{D}$. The effect of using an approximation of (\ref{eq-theo-weight}) based on a finite number of subsamples of $\c{D}$ is ignored here, but a detailed investigation by \citet{mentch_quantifying_2015} and \citet{wager_condence_2014} suggest using $B = O(n)$. 

The set of weights $\eset{w_i(\cdot; \xi, D)}_{i=1}^n$ given in (\ref{eq-tree-weight}) can be used to define the contribution of a single tree to the forest objective function
\begin{equation}\label{eq-theo-tree}
  T(\omega, x; \xi, D) = \sum_{i=1}^n w_i(x; \xi, D) d(\omega, Y_i)^2.
\end{equation}
The tree-level weights are then combined to form the theoretical random forest weights
\begin{equation}\label{eq-theo-weight}
  w_i^\star(x) = \binom{n}{s}^{-1} 
  \sum_{\substack{D \subset \c{D}\\\abs{D} = s}}
  \mathbb{E}_{\xi}\left[w_i(x; \xi, D)\right].
\end{equation}
Based on these weights, we can define the theoretical objective function and regression function studied in this section
\begin{equation}\label{eq-theo-rf}
  M_n^\star(\omega; x) = \sum_{i=1}^n w_i^\star(x)d(\omega, Y_i)^2
  \qquad\t{and}\qquad
  m_n^\star(x) = \argmin_{\omega \in \Omega} M_n^\star(\omega; x).
\end{equation}

We further consider that the forest is implemented in a way that satisfies Specification 1 of \citet{athey_generalized_2019}: All trees in the forest are symmetric, meaning that the order of the input is not relevant to the fitting of each tree. The splits are balanced in that each split separates the observations on two subsets, each with a proportion of at least $\alpha >0$ of the parent. Additionally, each tree is grown to depth $k$, for some $k \in \N$ and each leaf contains between $k$ and $2k - 1$ observations. The forest is honest as described in \citet{wager_estimation_2018}, and each tree is built on a subsample of the data of size $s < n$ satisfying $s / n \rightarrow 0$ and $s \rightarrow \infty$.

The following assumption is common in M-estimation, see for instance Assumption (P0) in \citet{muller_peter_2016}, and is used in most proofs of consistency of Fréchet regression methods.
\begin{assumption}[Existence of a population solution]\label{ass-ex-pop} For every $x \in [0,1]^d$, the conditional Fréchet regression function $m(x)$ given in (\ref{eq-cond-fr}) exists and is a \textit{well-separated} solution to the conditional Fréchet function $M(\cdot; x)$, that is for every $\varepsilon > 0$,
  \begin{equation*}
    \inf_{\omega: d(\omega, m(x)) > \varepsilon} M(\omega; x) > M(m(x); x).
  \end{equation*}
\end{assumption}

We further require the conditional Fréchet function to be Lipschitz continuous in $x$. This assumption is also common in Fréchet regression, and is used here to exploit the shrinkage of the tree leaves to show the convergence of the approximate Fréchet function defined by the random forest.

\begin{assumption}[$x$-Lipschitz Fréchet function]\label{ass-lip} For any fixed $\omega \in \Omega$, the function $x \mapsto M(\omega; x)$ is Lipschitz continuous with Lipschitz contant $C(\omega)$ possibly depending on $\omega$.
\end{assumption}

As indicated, we first state a theorem saying that the objective function $M_n^\star(\cdot;x)$ defined by the random forest algorithm converges uniformly to the true Fréchet conditional Fréchet function $M$. This is a necessary element in consistency proofs of M-estimators and thus used to prove consistency of the random forest algorithm (Theorem \ref{thm-consistency}). Proofs are given in the appendix. 

\begin{theorem} \label{thm-unif}
  Under Assumption \ref{ass-lip}, the random forest objective function $M_n^\star(\cdot; x)$ converges uniformly in $\omega$ to $M(\cdot; x)$, that is, for all $x \in [0, 1]^d$
  \begin{equation*}
      \sup_{\omega \in \Omega} \abs{M_n^\star(\omega; x) - M(\omega; x)} = o_P(1).
  \end{equation*}
\end{theorem}


\begin{theorem} \label{thm-consistency}  
  Under Assumptions \ref{ass-ex-pop} and \ref{ass-lip}, the random forest regressor defined in (\ref{eq-theo-rf}) is a pointwise consistent estimator of the conditional Fréchet regression function, that is, for each $x \in [0,1]^d$,
  \begin{equation*}
    d(m_n^\star(x), m(x)) = o_P(1).
  \end{equation*}
\end{theorem}

\section{Numerical experiments}\label{sec_numerical}

In the following, we perform a simulation study in order to evaluate the benefit of our medoid approach (MRF) as compared to the random forest weighted local constant Fréchet regression (RFWLCFR) algorithm presented in \citet{qiu_random_2022}. The main difference between the two algorithms is the choice of splitting criterion, cf.\ Section~\ref{sec_splitting_criterion}, and we will compare them in terms of computational efficiency and prediction quality.

We study three scenarios with similar data generation processes in three different metric spaces with different computational costs associated with the computation of empirical Fréchet means. The first example is that of one-dimensional density functions over the real line $\R$ in the Wasserstein space. This space is isomorphic to a convex subspace of $L_2$ in which computations of distances and Fréchet means are straightforward. In a second example, we study data lying on a Riemannian manifold, the sphere $S^2 \subset \R^3$. In this example distances are straightforward to compute but Fréchet means can only be computed through a more expensive gradient descent algorithm. In a last example, we study functional data presenting phase variability through the study of their warping functions as described in \citet{srivastava_functional_2016}, in which both the evaluation of distances and of the Fréchet mean are computationally demanding.

In each of the three metric spaces, we generate datasets of different training sample sizes $N \in \eset{100, 200, 400}$ and covariate dimension $d \in \eset{2, 5, 10, 20}$. For each combination of a metric space, a training sample size and a covariate dimension, we generate $100$ random datasets with the desired properties. In order to generate a dataset with covariate dimension $d$, two parameters $\alpha \in \R$ and $\beta \in \R^d$ are sampled with independent standard normal components. The Fréchet regression function is then generated from a single index regression model
\begin{equation}\label{eq-eta}
    m(x) = g(\eta)\qquad\t{with}\qquad\eta = \alpha + (x - 0.5)^\top \frac{\beta}{\sqrt{d}}
\end{equation}
for some function $g : \R \rightarrow \Omega$, which is fixed and specific to each metric space. The scaling $1 / \sqrt{d}$ has been chosen to ensure that the distribution of $\eta$ does not change too much as the number of parameters increases, thus keeping the different simulation settings comparable. Each observation $(X_i,Y_i)$ is finally constructed by first sampling $X_i \sim \t{Unif}[0,1]^d$, computing the conditional expectation $m(X_i)$, and applying an independently sampled noise function $T_i : \Omega \rightarrow \Omega$, $Y_i = T_i(m(Y_i))$. Note that $m$ varies is fixed within each dataset but varies between datasets. 

In a given experiment, we fit the MFR and RFWLCFR algorithms to the same training dataset and evaluate them on an independently generated test dataset of size $N_\t{test} = 100$ with the same parameters $\alpha, \beta$ as in the training set. We measure the training time of each algorithm and evaluate the performance of the algorithm using the mean squared error (MSE). For a test set $\eset{(X_i, Y_i)}_{i=1}^{N_\t{test}}$ and a fitted regressor $\hat m$, the MSE of the algorithm is given by
\begin{equation*}
    \t{MSE}(\hat m) = \frac{1}{N_\t{test}}\sum_{i=1}^{N_\t{test}} d(\hat m(X_i), m(X_i))^2.
\end{equation*}
The focus of our comparison is the runtime of the algorithms, while 
mean squared errors are used to assess whether the approximation in MRF has a negative impact on the quality of the predictions.
The simulations show the MFR provides significant performance gains over RFWLCFR, yet without compromising the quality of the predictions, and therefore enables us to analyze larger datasets in terms of sample size and dimension.

All simulations and analysis are done in Python. The implementation of both methods, together with the code to reproduce the results is available at \href{https://github.com/matthieubulte/MetricRandomForest}{https://github.com/matthieubulte/MetricRandomForest}.

\subsection{Univariate distributions}

We consider the 2-Wasserstein metric space $\Omega = W_2 = W_2(\R)$, the space of probability distributions over the real line with finite second moment endowed with the 2-Wasserstein metric $d_{W_2}$. Specifically, for $\mathbb{P}, \mathbb{Q} \in W_2$ with cumulative distribution functions $F_\mathbb{P}$ and $F_\mathbb{Q}$, the 2-Wasserstein distance between $P$ and $Q$ is given by
\begin{equation*}
    d_{W_2}(\mathbb{P}, \mathbb{Q}) = d_{L_2[0,1]} (F_\mathbb{P}^{-1}, F_\mathbb{Q}^{-1})= \sqrt{\int_0^1 \abs{F_\mathbb{P}^{-1}(u)- F_\mathbb{Q}^{-1}(u)}^2 \d u}.
\end{equation*}


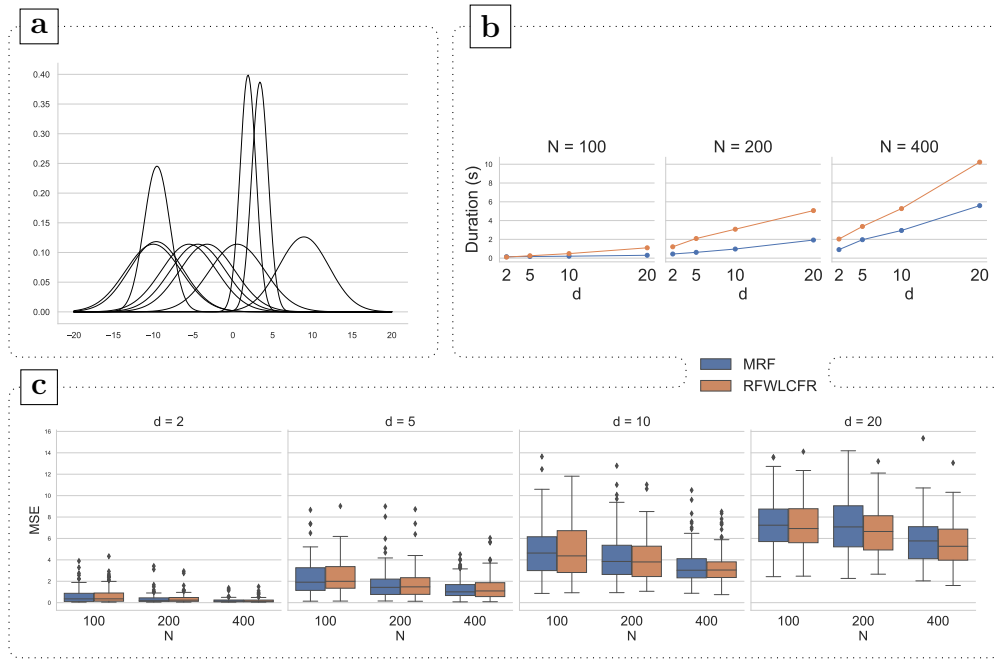
\begin{figure}[ t]
    \begin{tikzpicture}
        \node (A) at (-1.3, 1) {\includegraphics[width=6cm]{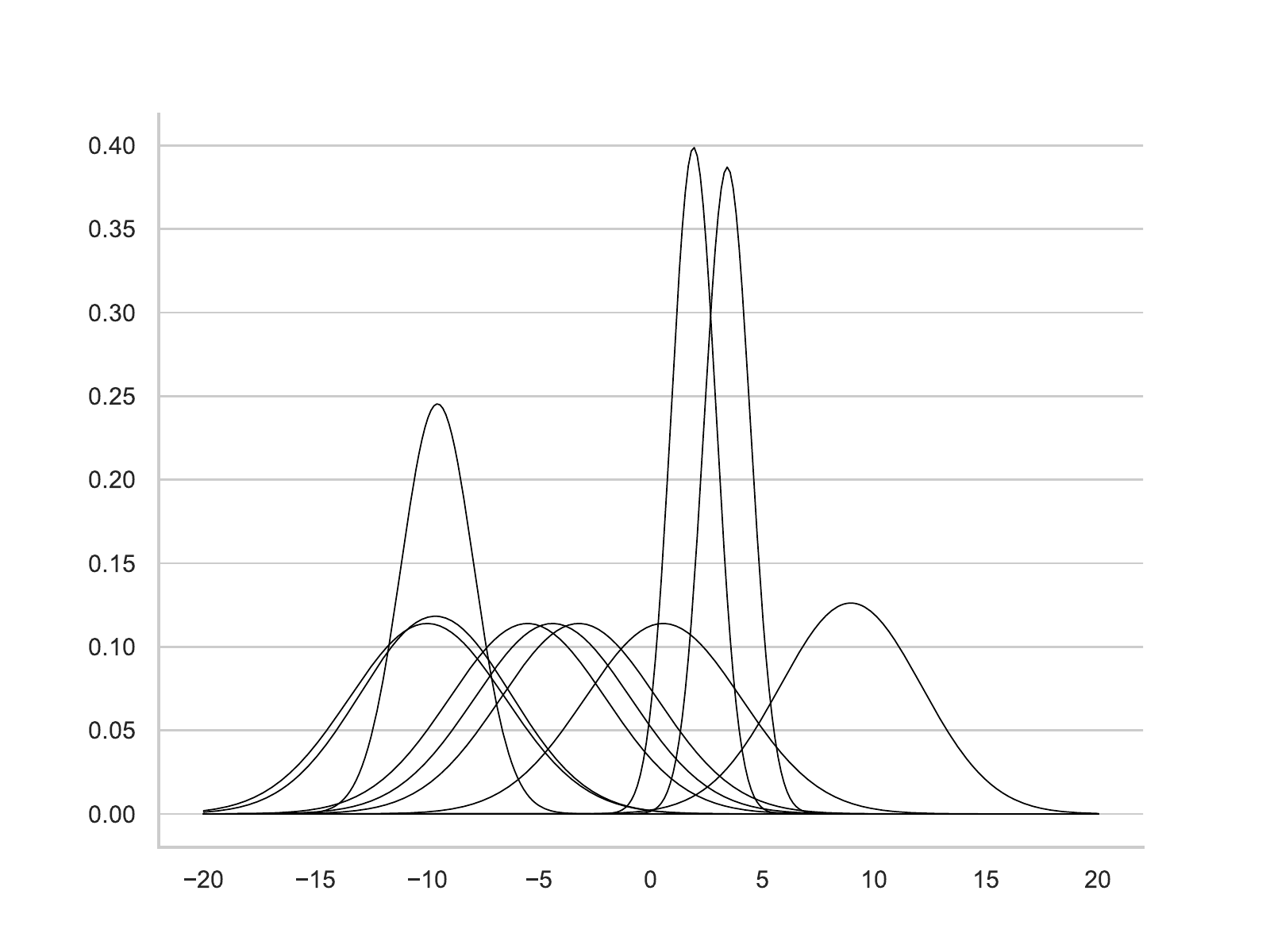}};
        \node (B) at (5.35, 0.6) {\includegraphics[width=7.2cm]{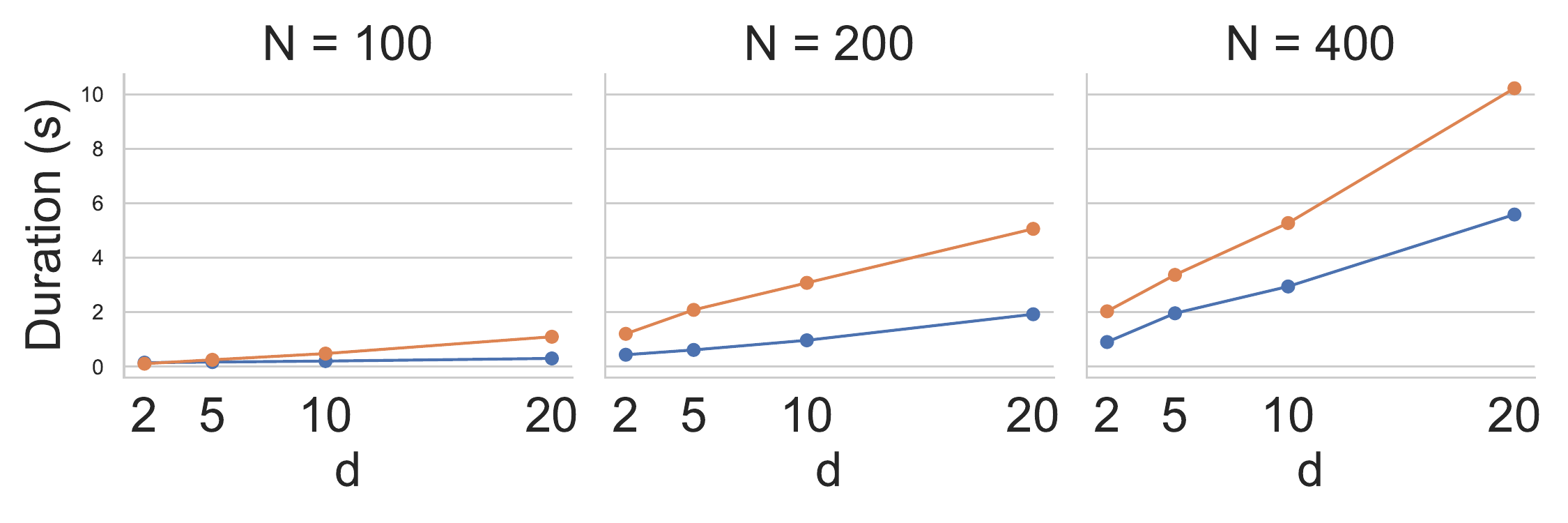}};
        \node (D) at (2.7, -3.5) {\includegraphics[width=13.5cm]{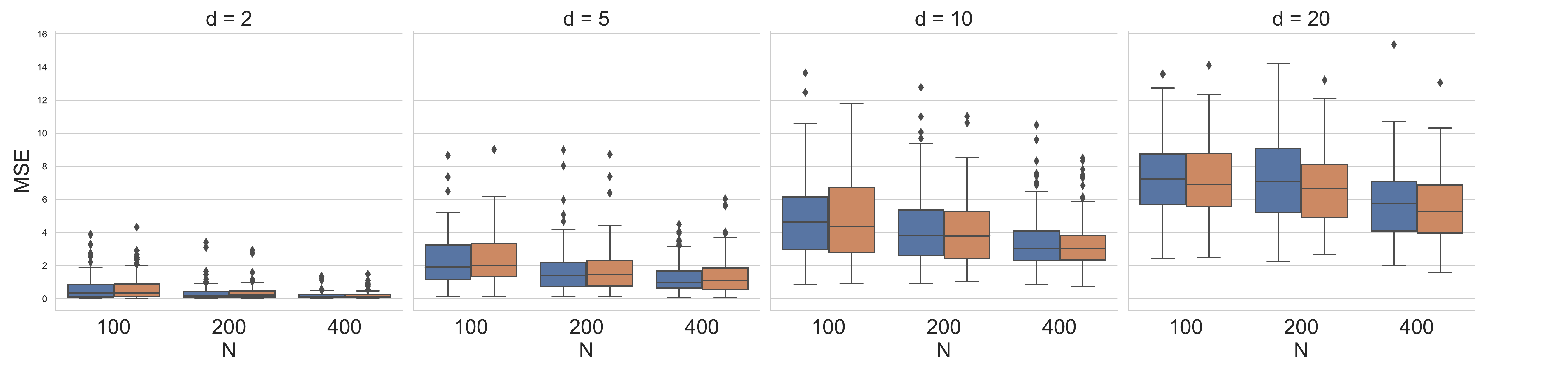}};

        \node (L) at (5.75, -1.45) {\includegraphics[width=2cm]{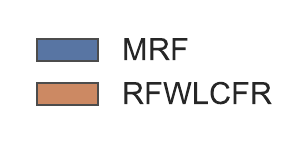}};
        
        \draw[dotted, rounded corners=4pt] (-4.2, 3.2) -- (1.5, 3.2) -- (1.5,-1.2) -- (-4.2,-1.2) -- cycle;
        \node[fill=white, draw] (aa) at (-3.8, 3.2) {\textbf{a}};

        \draw[dotted, rounded corners=4pt] (1.7, 3.2) -- (9, 3.2) -- (9, -1.2) -- (6.7, -1.2) -- (6.7, -1.6) -- (9, -1.6) -- (9, -5.2) -- (-4.2, -5.2) -- (-4.2, -1.6) -- (4.8, -1.6) -- (4.8, -1.2) -- (1.7, -1.2) -- cycle;
                
        \node[fill=white, draw] (bb) at (2.2, 3.2) {\textbf{b}};
        \node[fill=white, draw] (cc) at (-3.8, -1.6) {\textbf{c}};
    \end{tikzpicture}
    \caption{(a) shows the mean function $m(X)$ for 10 sampled values of $X$. (b) and (c) compare the runtime and MSE, respectively, of the MRF and RFWLCFR algorithms in Wasserstein regression for $N \in \eset{100,200,400}$ and $d \in \eset{2, 5, 10, 20}$.}
    \label{fig:wasserstein}
\end{figure}

As shown in existing literature on Wasserstein regression \citep{petersen_frechet_2019, ghodrati2022distribution}, the solution to the weighted Fréchet problem given via the sample $\eset{\mathbb{P}_i}_{i=1}^n \subset W_2$ and weights $\eset{w_i}_{i=1}^n$ can be found in a two-steps procedure: a pointwise weighted average of the sample inverse cumulative distribution functions is first computed and then projected onto the space of non-decreasing functions. This results in the following optimization problem
\begin{equation*}
    \argmin_{\mathbb{Q} \in W_2} \sum_{i=1}^n w_i d_{W_2}(\mathbb{Q}, \mathbb{P}_i)^2 = \argmin_{\mathbb{Q} \in W_2} \norm{F_{\mathbb{Q}}^{-1}(\cdot) - \sum_{i=1}^n w_i F_{\mathbb{P}_i}^{-1}(\cdot)}_{L_2[0,1]}^2
\end{equation*}
In practice, each distribution is represented by a vector $y_i = (y_{i1}, \ldots, y_{iM})$ of evaluations of its inverse cumulative distribution function on an equispaced grid $\eset{u_m}_{m=1}^M$ on $[0, 1]$. The weighted Fréchet mean $\bar y$ is then given by
\begin{equation*}
    F^{-1}_{\bar y} = \argmin_{ \substack{v \in \mathbb{R}^M\\ v_1 \leq \cdots \leq v_M} } \sum_{i=1}^n w_i \sum_{m=1}^M \abs{v_m -  y_{im} }^2.
\end{equation*}

In our application, we chose $M = 100$. As highlighted in \citet{ghodrati2022distribution}, this is an isotonic regression problem which can be efficiently solved using scikit-learn's implementation based on Pool Adjacent Violators Algorithm (PAVA), see \citet{pedregosa_scikit-learn_2011,best_active_1990}.

Following the example in \citet{petersen_frechet_2019} for data generation, the conditional mean in this experiment maps a covariate vector to a normal distribution with variable mean and variance parameters,
\begin{align*}
    F^{-1}_{m(x)}(u) = \mu(\eta) + \sigma(\eta) \Phi^{-1}(u).
\end{align*}
Specifically, we use the following functions for the mean and variance components,
\begin{equation*}
    \mu(\eta) = \eta \qquad\t{and}\qquad \sigma(\eta) = \sigma_0 + \gamma \t{logit}^{-1}(\eta),
\end{equation*}
with $\sigma_0 = 1$ and $\gamma = 2.5$. The random responses are generated by adding random noise to the normal quantiles as described in \citet{panaretos_amplitude_2016}.  Given a continuous and non-decreasing random noise function $\varepsilon : \R \rightarrow \R$, the noise map $T : \Omega \rightarrow \Omega$ is defined via the composition of $\varepsilon$ with the quantile function of $\P$, giving $F^{-1}_{T(\P)} = \varepsilon \circ F^{-1}_\P$. The noise functions are given by $\varepsilon(x) = x - \sin\left(\pi k x\right)/ \abs{\pi k}$, where $k$ is an integer frequency uniformly sampled from $\eset{-4, \ldots, 4} \backslash \eset{0}$. The observation $Y$ given $X$ are thus given by $F^{-1}_{Y} = \varepsilon \circ F^{-1}_{m(X)}$. Note that adding randomness via composition with a non-linear function results in random distributions outside the class of normal distributions.

Figure \ref{fig:wasserstein} compares the time performance and MSE of the methods. Due to the isomorphism of $W_2$ to the subspace of $L_2[0,1]$ of quantile functions, the computation of Fréchet means is rather straightforward. However, the $2$-means procedure adds an overhead to the RFWLCFR methodology which makes the overall computation slower than our MRF. The methods have similar MSE distributions in all setups showing that the medoid approximation in the MRF does not affect the quality of the predictions.

\subsection{Spherical data}

\begin{figure}[h]
    \begin{tikzpicture}
        \node (A) at (-1.3, 1) {\includegraphics[width=6cm]{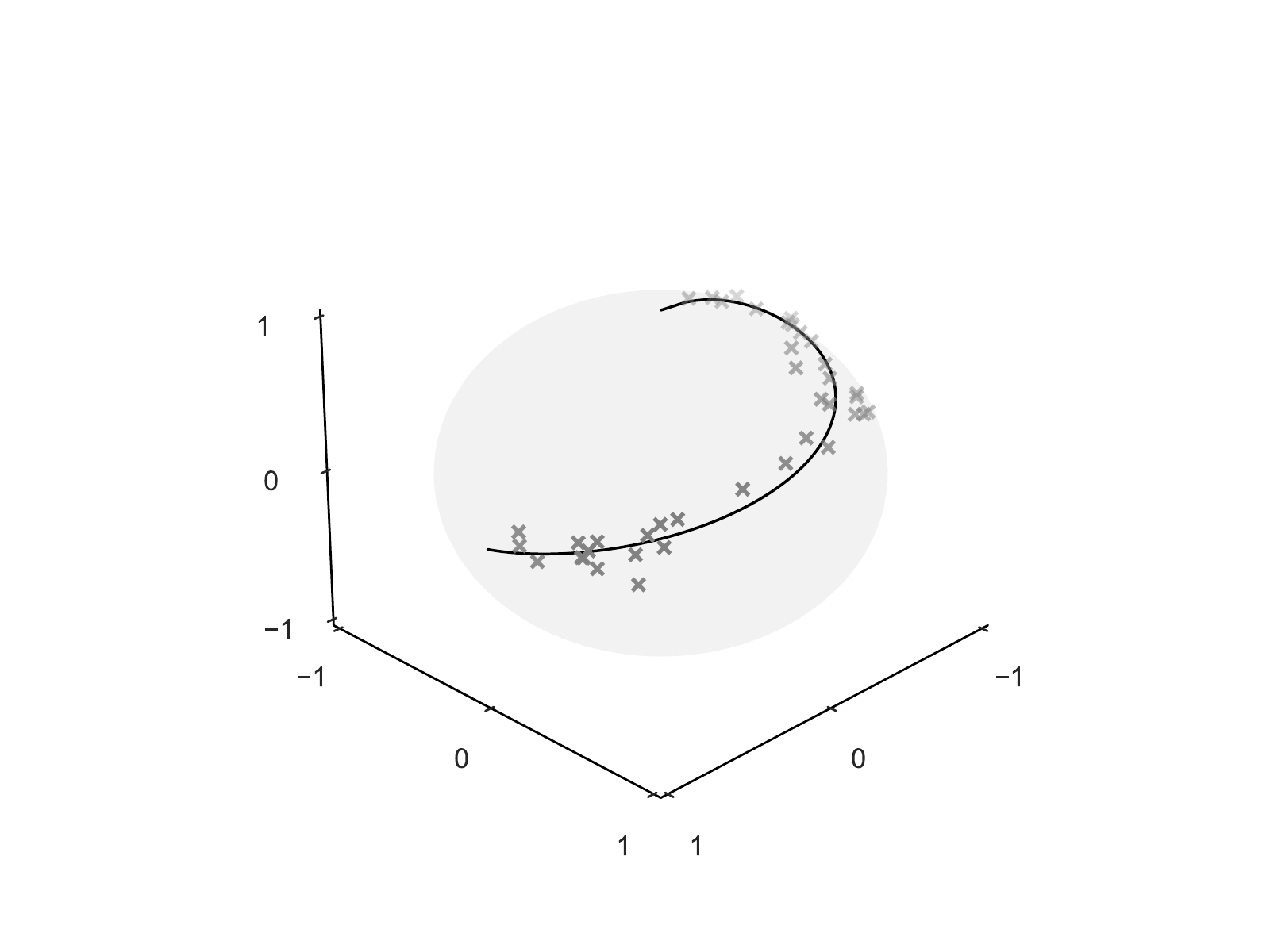}};
        \node (B) at (5.35, 0.6) {\includegraphics[width=7.2cm]{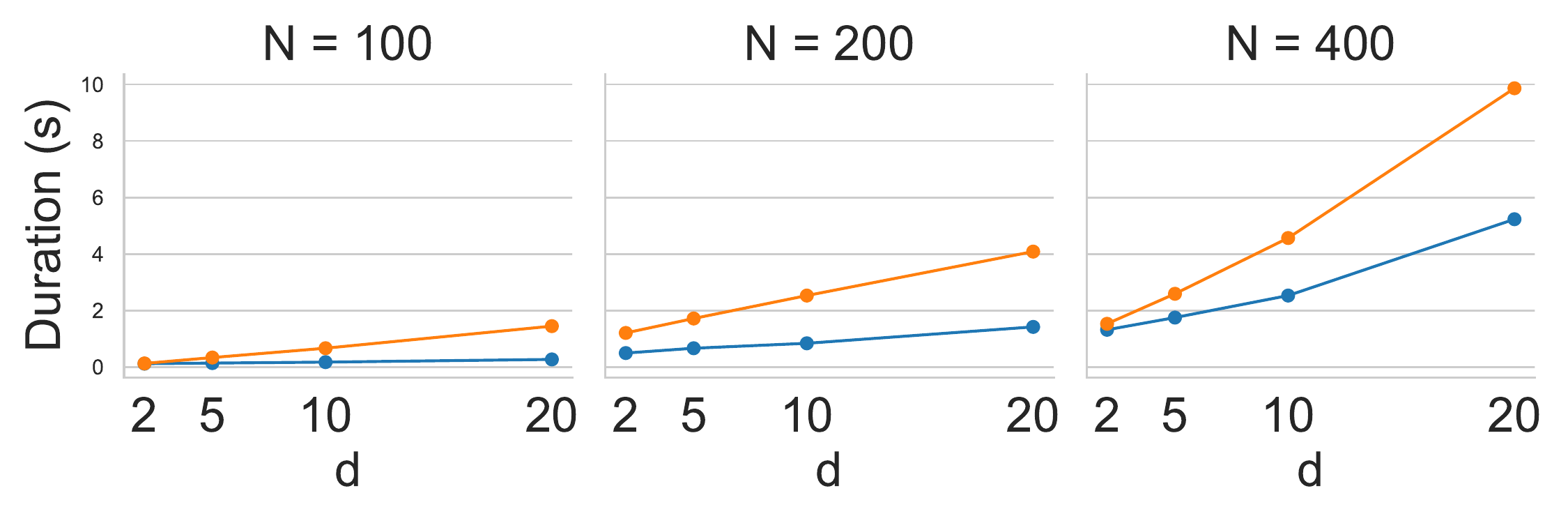}};
        \node (D) at (2.7, -3.5) {\includegraphics[width=13.5cm]{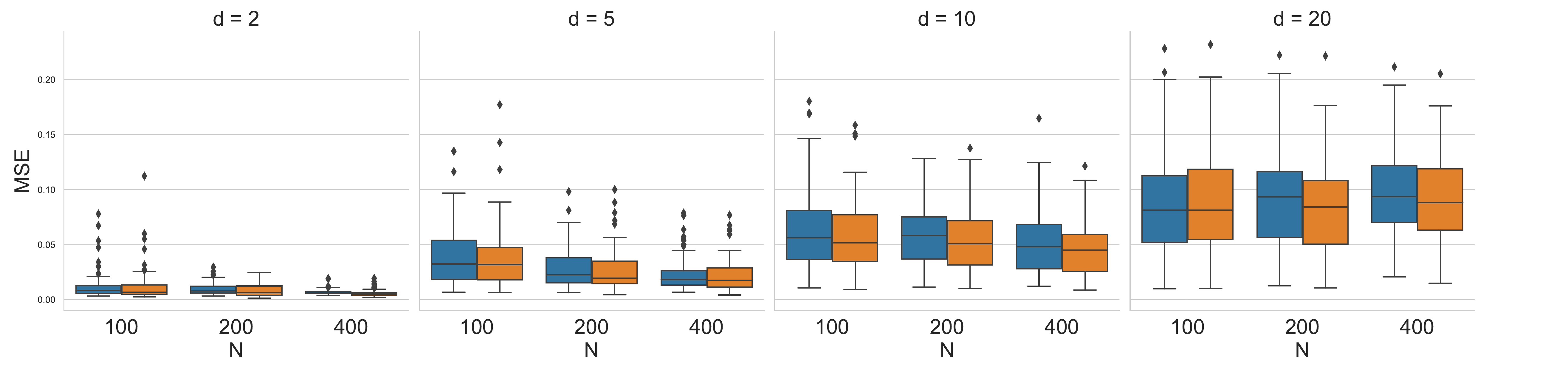}};

        \node (L) at (5.75, -1.45) {\includegraphics[width=2cm]{plots/legend.pdf}};
        
        \draw[dotted, rounded corners=4pt] (-4.2, 3.2) -- (1.5, 3.2) -- (1.5,-1.2) -- (-4.2,-1.2) -- cycle;
        \node[fill=white, draw] (aa) at (-3.8, 3.2) {\textbf{a}};

        \draw[dotted, rounded corners=4pt] (1.7, 3.2) -- (9, 3.2) -- (9, -1.2) -- (6.7, -1.2) -- (6.7, -1.6) -- (9, -1.6) -- (9, -5.2) -- (-4.2, -5.2) -- (-4.2, -1.6) -- (4.8, -1.6) -- (4.8, -1.2) -- (1.7, -1.2) -- cycle;
                
        \node[fill=white, draw] (bb) at (2.2, 3.2) {\textbf{b}};
        \node[fill=white, draw] (cc) at (-3.8, -1.6) {\textbf{c}};
    \end{tikzpicture}
    \caption{(a) shows the mean curve on the sphere (solid line) as a function of $\nu$ in (\ref{eq-sphere-mean}) together with random observations (crosses). (b) and (c) compare the runtime and MSE, respectively, of the MRF and RFWLCFR algorithms for regression on the sphere $\mathbb{S}^2$ for $N \in \eset{100,200,400}$ and $d \in \eset{2, 5, 10, 20}$.}
    \label{fig:sphere}
\end{figure}

We consider the example where $\Omega = \mathbb{S}^q \subset \R^{q+1}$ is the $q$-dimensional sphere equipped with the geodesic distance
\begin{equation*}
    d_{\mathbb{S}^q}(x, y) = \arccos \langle x, y \rangle.
\end{equation*}
In this case, computing a weighted Fréchet mean is less straightforward and one must resort to optimization method on manifolds. We do this using the trust region algorithm implemented in Pymanopt \citep{townsend_pymanopt_2016}.

We focus on data lying on the sphere $\mathbb{S}^2 \subset \R^3$ generated as follows. First, the mean function is given by the single index model described earlier, where we first transform $\eta$ to $\nu \in (0,1)$ via $\nu = \t{logit}^{-1}(\eta)$, and define
\begin{equation}\label{eq-sphere-mean}
    g(\eta) = \left(\sqrt{1 - \nu^2} \cos(\pi \nu), \sqrt{1 - \nu^2} \sin(\pi \nu), \nu\right).
\end{equation}
Given a random $X \sim \t{Unif}[0,1]^d$ and its associated mean function $m(X)$, the response $Y \in \mathbb{S}^2$ is generated by transforming a bivariate random vector $U$ from the tangent space $T_{m(X)}\mathbb{S}^2$ at $m(X)$ through the corresponding exponential map,
\begin{equation*}
    Y = \t{Exp}_{m(X)}(U) = \cos(\norm{U})m(X) + \sin(\norm{U})\frac{U}{\norm{U}}.
\end{equation*}
The random noise $U$ is sampled from a bivariate normal distribution with independent components and variance $\sigma^2 = 0.1$, that is $U \sim \mathcal{N}(0, \sigma^2 \mathbbm{1}_{2 \times 2})$.

The results of the experiment are displayed in Figure \ref{fig:sphere}. The time performance profile is similar to the one in the Wasserstein experiment with MRF still providing a clear advantage over the RFWLCFR in terms of runtime. In this example the higher cost of computing Fréchet means does not seem to translate into further runtime improvements of MFR over RFWLCFR. The two methods again yield estimators with similar error distribution.

\subsection{Space of warping functions}\label{subsec_fd_phase}

We now consider the study of functional data displaying phase variability, with observations lying in a non-linear subset of $L_2[0,1]$. As described in \citet{tucker_generative_2013}, and more generally \citet{srivastava_functional_2016}, the variability of functional data can be decomposed in phase and amplitude components. Intuitively, phase variation corresponds to variation on the $x$ axis and amplitude variation corresponds to variation along the $y$ axis. 

Let $\Gamma$ be the set of boundary preserving diffeomorphisms on $[0, 1]$, that is, $\Gamma = \eset{ \gamma : [0,1] \rightarrow [0,1] \mid \gamma(0) = 0, \gamma(1) = 1, \gamma \t{ is a diffeomorphism}}$. For a given function $y_0 \in L^2[0,1]$, we define the orbit of $y_0$ as the set of functions that only differ from $y_0$ in their phase component, $[y_0] = \eset{ y_0 \circ \gamma \mid \gamma \in \Gamma }$. While $y_0$ is not unique in this notation, whenever the starting point $y_0$ of an orbit $[y_0]$ is fixed in a context, we call it \textit{template function}. 

We consider a distribution for $Y$ in which the observations are all part of the same orbit $\Omega = [y_0]$. Such spaces of equivalence classes, which are all equivalent to the space of warping functions $\Gamma$, form a non-linear infinite-dimensional manifold which is difficult to analyze and work with when considering the $L_2$ distance. 

Instead, a warping function can be represented by its square-root velocity function $\psi_\gamma = \sqrt{\dot \gamma}$. Since the warpings $\gamma$ are boundary preserving diffeomorphisms on $[0,1]$ we have that $\dot \gamma$ exists and is positive, making $\psi_\gamma$ a well defined quantity. Furthermore $\gamma(0) = 0$ and $\gamma(1) = 1$, hence $\norm{\psi}_{L_2}^2 = \int_0^1 \psi(t)^2 \d t = \int_0^1 \dot \gamma(t) \d t = 1$ making $\psi$ an element of the $L_2[0,1]$ sphere, $\mathbb{S}^\infty = \eset{ f \in L_2[0,1] \mid \norm{f}_{L_2[0,1]} = 1}$, and more specifically on the positive hemisphere of $\mathbb{S}^\infty$ since $\psi > 0$. This allows to define a distance on $\Gamma$ via the intrinsic distance on $\mathbb{S}^\infty$. Let $\gamma_1, \gamma_2$ be two warpings with square-root velocity functions $\psi_1, \psi_2$, then the distance between $\gamma_1$ and $\gamma_2$ is given by
\begin{equation*}
    d_\Gamma(\gamma_1, \gamma_2) = d_{\mathbb{S}^\infty}(\psi_1, \psi_2) = \arccos \langle \psi_1,\psi_2 \rangle_{L_2}.
\end{equation*} 
See Chapter 4 of \citet{srivastava_functional_2016} for more details on this metric space. 
Let now $y_1, y_2 \in [y_0]$ be two functions within a same orbit, then these functions can be represented by their warpings from the template $y_0$, denoted $\gamma_1, \gamma_2$, with $y_i = y_0 \circ \gamma_i$ and we can define
\begin{equation*}
    d_{[y_0]}(y_1, y_2) = d_\Gamma(\gamma_1, \gamma_2).
\end{equation*}
The computation of quantities required to evaluate distances and the Fréchet mean are based on implementations provided in scikit-fda \citep{ramos_carreno_carlos_2019_7573921} and fdasrsf \citep{tucker_generative_2013}.

\begin{figure}[t!]
    \begin{tikzpicture}
        \node (A) at (-1.3, 1) {\includegraphics[width=6cm]{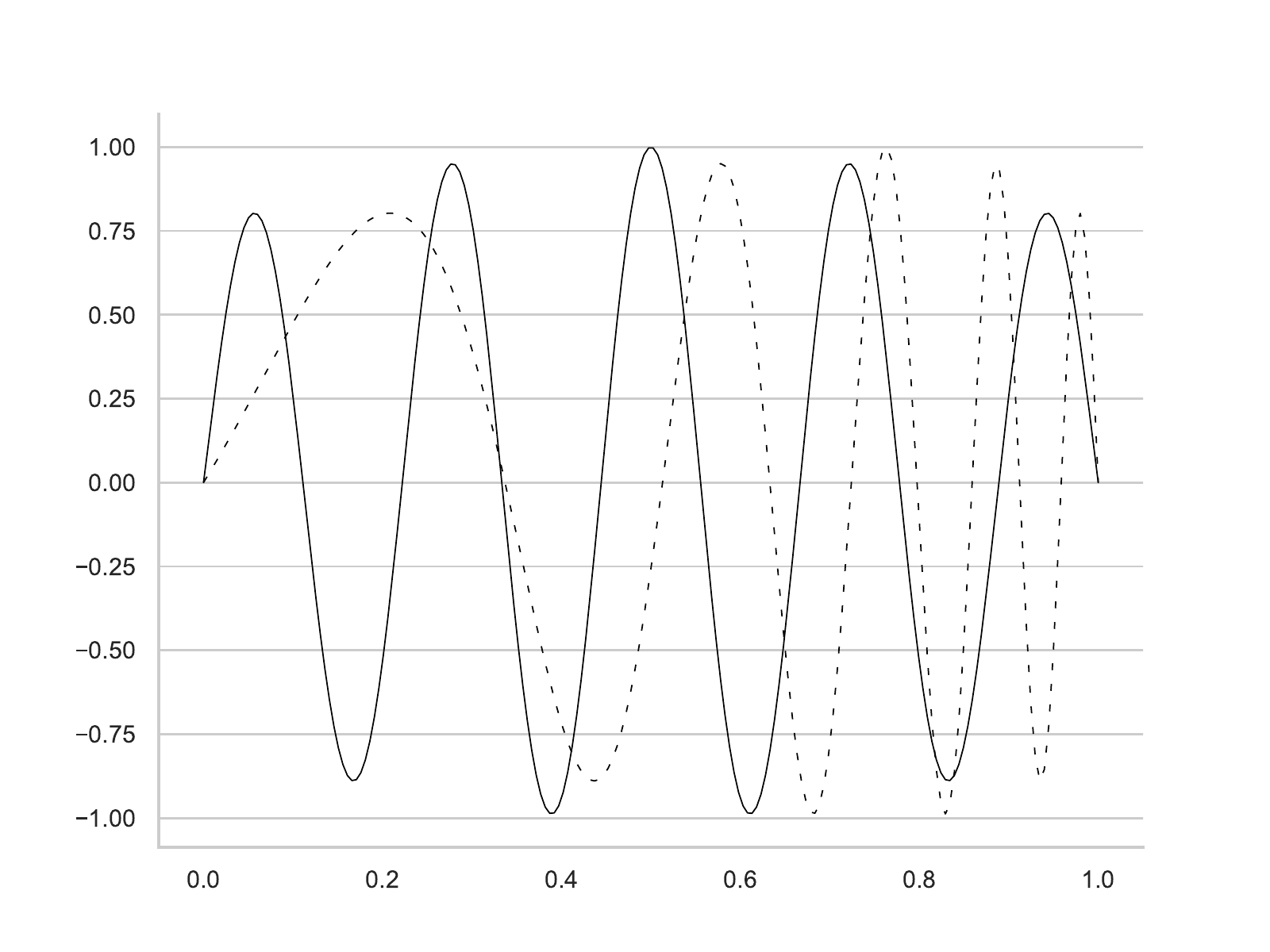}};
        \node (B) at (5.35, 0.6) {\includegraphics[width=7.2cm]{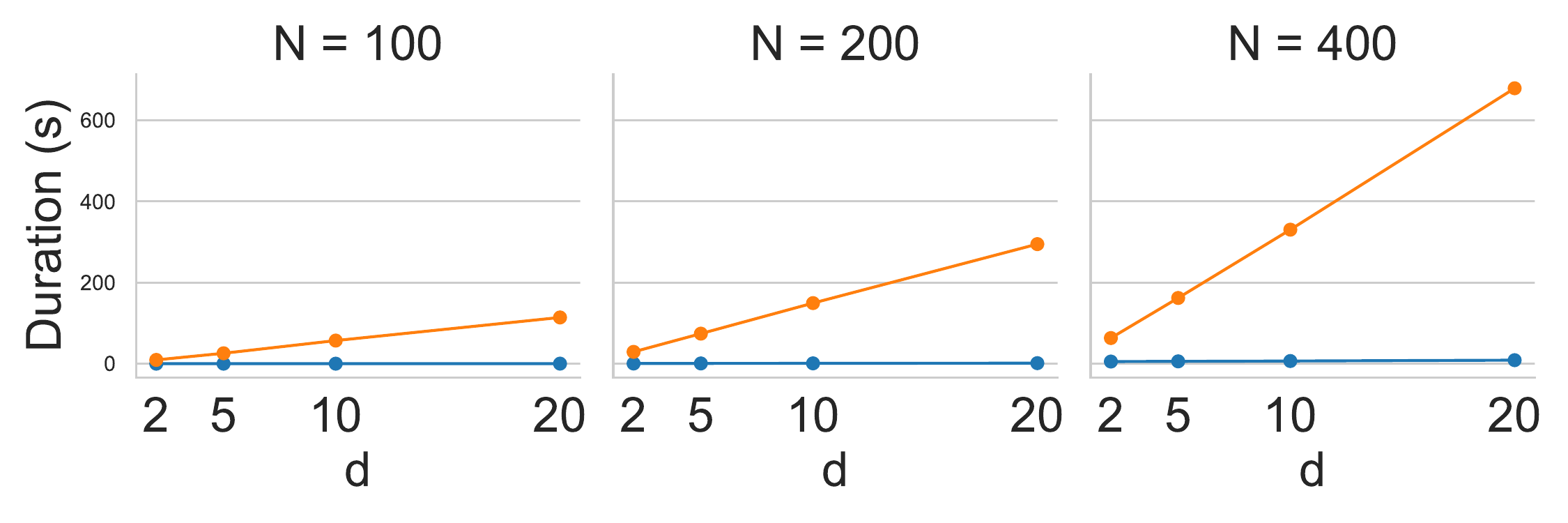}};
        \node (D) at (2.7, -3.5) {\includegraphics[width=13.5cm]{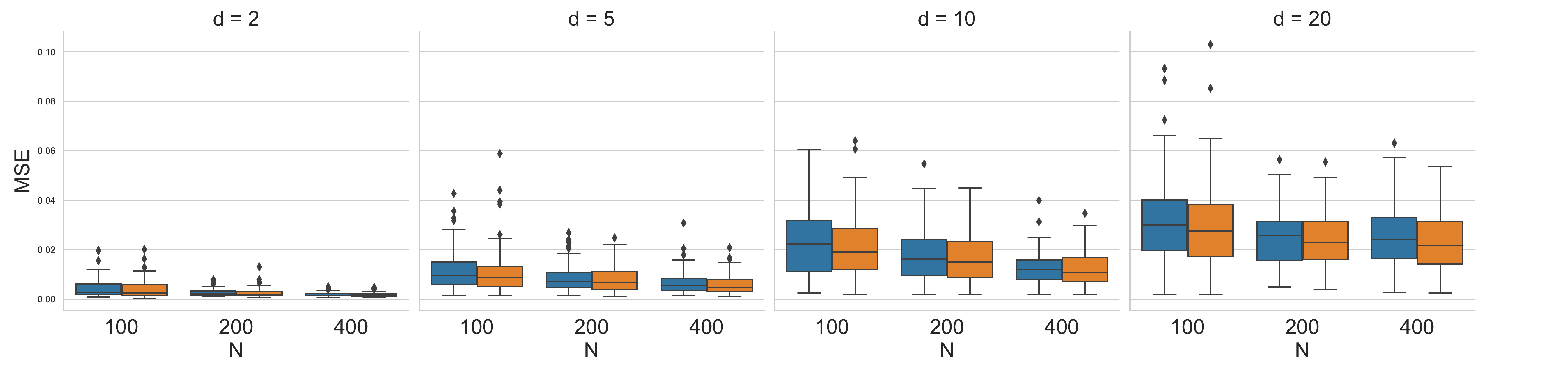}};

        \node (L) at (5.75, -1.45) {\includegraphics[width=2cm]{plots/legend.pdf}};
        
        \draw[dotted, rounded corners=4pt] (-4.2, 3.2) -- (1.5, 3.2) -- (1.5,-1.2) -- (-4.2,-1.2) -- cycle;
        \node[fill=white, draw] (aa) at (-3.8, 3.2) {\textbf{a}};

        \draw[dotted, rounded corners=4pt] (1.7, 3.2) -- (9, 3.2) -- (9, -1.2) -- (6.7, -1.2) -- (6.7, -1.6) -- (9, -1.6) -- (9, -5.2) -- (-4.2, -5.2) -- (-4.2, -1.6) -- (4.8, -1.6) -- (4.8, -1.2) -- (1.7, -1.2) -- cycle;
                
        \node[fill=white, draw] (bb) at (2.2, 3.2) {\textbf{b}};
        \node[fill=white, draw] (cc) at (-3.8, -1.6) {\textbf{c}};
    \end{tikzpicture}
    \caption{(a) shows the template curve $y_0$ (solid) and a randomly sampled curve (dashed). (b) and (c) compare the runtime and MSE, respectively, of the MRF and RFWLCFR algorithms for regression in $\Gamma$ for $N \in \eset{100,200,400}$ and $d \in \eset{2, 5, 10, 20}$.}
    \label{fig:phase}
\end{figure}

Next, we turn to the data generating process for our experiment. The template function is chosen as $y_0(u) = (1 - (u - 0.5)^2) \sin(9 \pi u)$, see the solid curve in panel (a) of Figure \ref{fig:phase}. For a given random vector of covariates $X$, one obtains the linear component $\eta$ as defined in (\ref{eq-eta}) and the mean warping response is $\gamma$ is given for all $u \in [0,1]$ by
\begin{equation*}
    \gamma(u) = (\exp(4au) - 1) / (\exp(4 a) - 1),
\end{equation*}
where $a = 3(\t{logit}^{-1}(\eta) - 0.5)$. The warping is perturbed by adding noise to its square-root velocity function $\psi \in \mathbb{S}^\infty$ by sampling a random Gaussian process $V$ with exponential covariance kernel from the tangent space $T_{\psi} \mathbb{S}^\infty$, and applying the exponential map at $\psi$,
\begin{equation*}
    \t{Exp}_{\psi} V = \cos(\norm{V}) \psi + \sin(\norm{V})\frac{V}{\norm{V}}.
\end{equation*} 
The response $Y$ is then obtained by warping the template $y_0$ with the warping function obtained by taking the inverse square-root velocity transform of $\t{Exp}_{\psi} V$. An example of a sample is displayed in panel (a) of Figure \ref{fig:phase} (dashed curve). 

The results of the experiment are displayed in Figure \ref{fig:phase}. Computing a single empirical Fréchet mean can take up to several seconds and computing Fréchet variance is thus a very expensive operation. The benefit of the 2-means approach implemented in RFWLCFR is dominated by the high cost of computing Fréchet variances. On the other hand, the runtime of MFR remains low in all scenarios. Runtime ratios are 1 to 50 in the $N = 100$ and $d = 20$ scenario and 1 to 30 for $N = 200$ and $d = 20$. Similarly to the previous two experiments, the runtime gain does not impact the fitting of the forest and the MSE distributions stay similar for the two methods.
\section{Discussion} \label{sec-conclusion}
Defining non-parametric regression method for general metric spaces that are performant and work well in complex setups is a challenge. We have proposed a new version of the random forest Fréchet regression that allows practitioners to perform a random forest regression in high-dimensional problems and in metric spaces where computation of Fréchet means is expensive. We achieved this by replacing the classic computation of the Fréchet means in random forest splits with a computation based on the Fréchet medoid. We proved both the consistency of the random forest Fréchet regression as well as the consistency of our medoid-based approximation. Furthermore, a set of numerical experiments in metric spaces of different characteristics demonstrated important reduction in the time required to fit a random forest without showing any impact on the mean squared error of the method.

One possible extension of this work could be to study rates of convergence, at best based on assumptions on the data generating process only. This, combined with a further investigation of the medoid approximation, could refine the understanding of the asymptotic behaviour of the metric random forest. 
Our work also paves the way to exploring further approximations to random forest regression in metric spaces. This could be through the use of algorithmic improvements generally applicable to random forests, such as the improvement proposed in \citet{NEURIPS2022_08857467}, or new approximations specific to the use of metric spaces. One could for instance study the use of approximate distance computations or other approximations to the Fréchet variance, and incorporate such approximations in the theoretical analysis as well as in the implementation. 

\section{Acknowledgement} \label{sec-ack}

This work has received funding from the European Union’s Horizon 2020 research and innovation program under the Marie Skłodowska-Curie grant agreement No 956107, "Economic Policy in Complex Environments (EPOC)".

\appendix

\section{Proofs} \label{sec-proofs}
\subsection{Statement of useful theorems}\label{sub-useful}

We start by stating two theorems that will be used in the proof of consistency. First, Corollary 2.2 of \citet{newey_uniform_1991} provides the conditions under which pointwise convergence of a sequence of random functions can be translated to uniform convergence.
\begin{theorem} \label{thm-newey}
    Let $(\Omega, d)$ be a compact metric space, $M_n : \Omega \rightarrow \R$ be a sequence of random functions and $M : \Omega \rightarrow \R$ be a fixed function such that
    \begin{enumerate}
        \item (Continuity): $M$ is continuous.
        \item (Stochastic equicontinuity): There exist some $B_n = O_P(1)$ such that for all $\omega, \omega' \in \Omega, \abs{M_n(\omega) - M_n(\omega')} \leq B_n d(\omega, \omega')$.
        \item (Pointwise convergence): For every $\omega \in \Omega$, $M_n(\omega) - M(\omega) = o_P(1)$.
    \end{enumerate}
    Then 
    \begin{equation*}
        \sup_{\omega \in \Omega} \abs{M_n(\omega) - M(\omega)} = o_P(1).
    \end{equation*}
\end{theorem}

Further, Theorem 5.7 of \citet{vaart_asymptotic_1998} can be used to translate uniform convergence of a sequence of objective functions to convergence of the minimizers.
\begin{theorem} \label{thm-vdv}
    Let $(\Omega, d)$ be a metric space, $M_n : \Omega \rightarrow \R$ be a sequence of random functions and $M : \Omega \rightarrow \R$ be a fixed function such that for every $\varepsilon > 0$,
    \begin{align*}
        \sup_{\omega \in \Omega} \abs{M_n(\omega) - M(\omega)} = o_P(1)\\
        \inf_{\omega: d(\omega_0, \omega) > \varepsilon} M(\omega) > M(\omega_0).
    \end{align*}
    Then any sequence of estimators $\hat\omega_n$ with $M_n(\hat\omega_n) \leq M_n(\omega_0) + o_P(1)$ converges in probability to $\omega_0$.
\end{theorem}
The first condition of this theorem states that $M_n$ converges uniformly to $M$ while the second condition states that there exists a well-separated minimizer $\omega_0$ of $M$.
\subsection{Consistency of the MRF estimator}

We start with a lemma regarding Lipschitz continuity of the distance and squared distance functions. 
\begin{lemma} \label{lem-lipschitz}
    Let $(\Omega, d)$ be a compact metric space and $\omega_0 \in \Omega$, then the maps $\omega \mapsto d(\omega_0, \omega)$ and $\omega \mapsto d(\omega_0, \omega)^2$ are Lipshitz continuous.
\end{lemma}
\begin{proof}
    Starting with the first map, let $\omega, \omega' \in \Omega$, then we have by the triangle inequality that
    \begin{align*}
        d(\omega_0, \omega) \leq d(\omega_0, \omega') + d(\omega', \omega)
        &\Rightarrow
        d(\omega_0, \omega) - d(\omega_0, \omega') \leq d(\omega', \omega)\\
        d(\omega_0, \omega') \leq d(\omega_0, \omega) + d(\omega', \omega)
        &\Rightarrow
        -\left(d(\omega_0, \omega) - d(\omega_0, \omega')\right) \leq d(\omega', \omega),
    \end{align*}
    and hence
    \begin{equation*}
        \abs{d(\omega_0, \omega) - d(\omega_0, \omega')} \leq d(\omega', \omega).
    \end{equation*}
    For the second map, we have
    \begin{align*}
        \abs{d(\omega_0, \omega)^2 - d(\omega_0, \omega')^2}
        &= \abs{d(\omega_0, \omega) + d(\omega_0, \omega')}\abs{d(\omega_0, \omega) - d(\omega_0, \omega')}\\
        &\leq 2 \diamO\,d(\omega, \omega'),
    \end{align*}
    where the first term was bounded by the diameter of $\Omega$ and we used the Lipshitz continuity of $\omega \mapsto d(\omega_0, \omega)$ for the second term.
\end{proof}

Next, we prove that the random forest objective function converges uniformly in $\Omega$ to the conditional variance. Recall definitions \eqref{eq-theo-rf} of $M_n^\star(\omega; x)$ and 
\eqref{eq-cond-fr} of $M(\omega; x)$.

\begin{proof}[Proof of Theorem \ref{thm-unif}]
    This proof is done by verifying that for a fixed $x \in [0,1]^d$, the maps $\omega \mapsto M_n^\star(\omega; x)$ and $\omega \mapsto M(\omega; x)$ respect the conditions of Theorem \ref{thm-newey}.
    
    \paragraph*{Continuity} This follows directly from Lemma \ref{lem-lipschitz} since for $\omega, \omega' \in \Omega$,
    \begin{align*}
        \abs{M(\omega; x) - M(\omega'; x)} 
        &= \abs{\expec{d(\omega, Y)^2 \mid X = x} - \expec{d(\omega', Y)^2 \mid X = x}} \\
        &\leq \expec{\abs{d(\omega, Y)^2 - d(\omega', Y)^2} \mid X = x} \\
        &\leq 2 \diamO\,d(\omega, \omega'),
    \end{align*}
    showing that $M$ is Lipschitz continuous and hence continuous.

    \paragraph*{Stochastic Equicontinuity} Similarly to the continuity of $M$, the stochastic equicontinuity of $M_n^\star$ is inherited. Let $\omega, \omega' \in \Omega$, 
    \begin{align*}
        \abs{M_n^\star(\omega; x) - M_n^\star(\omega'; x)}
        &= \abs{\sum_{i=1}^n w_i^\star(x) d(\omega, Y_i)^2 - \sum_{i=1}^n w_i^\star(x) d(\omega', Y_i)^2}\\
        &\leq \sum_{i=1}^n w_i^\star(x) \abs{d(\omega, Y_i)^2 - d(\omega', Y_i)^2}\\
        &\leq \sum_{i=1}^n w_i^\star(x) 2 \diamO\,d(\omega, \omega')\\
        &= 2 \diamO\,d(\omega, \omega').
    \end{align*}

    \paragraph*{Pointwise convergence} Let $\omega \in \Omega$ be fixed. We need to show that $M_n^\star - M = o_P(1)$. Recall the contribution of a single fitted tree (\ref{eq-theo-tree}) to the theoretical random forest objective function $M_n^\star$,
    \begin{equation*}
        T(\omega, x; \xi, D) = \sum_{i=1}^n w_i(x; \xi, D) d(\omega; Y_i)^2.
    \end{equation*}
    We simplify the notation and remove references to $x, \omega, \xi$ and $D$ in the remaining of this sub-proof. We now define the Hájek projection $\hajek{M_n^\star}$ of $M_n^\star$,
    \begin{equation*}
        \hajek{M_n^\star} = \expec{M_n^\star} + \sum_{i=1}^n \expec{M_n^\star \mid (X_i, Y_i)} - \expec{M_n^\star}.
    \end{equation*}
    This can be rewritten in terms of expectations of the individual trees as
    \begin{equation*}
        \hajek{M_n^\star} = \expec{T} + \frac{s}{n} \sum_{i=1}^n \expec{T \mid (X_i, Y_i)} - \expec{T}.
    \end{equation*}
    We proceed with the proof of pointwise convergence by showing the stronger result of $L_2$ convergence of $M_n^\star$, that is $\expec{\left(M_n^\star - M\right)^2} = o(1)$. We do this via the convergence of the Hájek projection of the forest weights $M_n$, in a similar fashion to the analysis in \citet{wager_estimation_2018}. The mean squared error $\expec{\left(M_n^\star - M\right)^2}$ can be upper bounded by the following sum in which each term is more amenable to analysis:
    \begin{equation*}
        \expec{\left(M_n^\star - \hajek{M_n^\star}\right)^2} + \expec{\left(\hajek{M_n^\star} - \expec{M_n^\star}\right)^2} + \left(\expec{M_n^\star} - M\right)^2.
    \end{equation*}
    
    By the specification of the construction of the random forest, we have that Lemma 7 in \citet{wager_estimation_2018} applies and we can bound the first summand as
    \begin{equation*}
        \expec{\left(M_n^\star - \hajek{M_n^\star}\right)^2} \leq \frac{s^2}{n^2} \Var{T}.
    \end{equation*}
    As shown in the proof of Theorem 5 from \citet{wager_estimation_2018}, the variance of a single tree is asymptotically bounded
    \begin{eqnarray*}
        \Var{T} \lesssim k \Var{d(\omega, Y)^2 \mid X = x} < \infty,
    \end{eqnarray*}
    where $k$ is the depth at which the trees are grown, as specified in Section \ref{sec-asymptotics}. Hence 
    \begin{equation*}
        \expec{\left(M_n^\star - \hajek{M_n^\star}\right)^2} \lesssim \frac{s^2}{n^2} = o(1)
    \end{equation*}
    
    For the second summand, we note that $\expec{M_n^\star} = \expec{\hajek{M_n^\star}}$ and hence the second term is the variance of the Hájek projection. We can again use the fact that the variance of a single tree is asymptotically bounded to obtain
    \begin{equation*}
        \Var{\hajek{M_n^\star}} = \frac{s^2}{n} \Var{\expec{T \mid Z_1}} \leq \frac{s}{n} \Var{T} \lesssim \frac{s}{n} = o(1)
    \end{equation*}

    Consider at last the third summand.
    Since $M_n^\star$ is constructed as the mean of identically distributed trees, we have that $\expec{M_n^\star} = \expec{T}$. We thus analyze instead the quantity
    \begin{equation*}
        \expec{T} - \expec{d(\omega, Y)^2 \mid X=x}.
    \end{equation*}
    As done in the proof of Theorem 3 in \citet{wager_estimation_2018}, we have by honesty that
    \begin{align*}
        &\expec{T} - \expec{d(\omega, Y)^2 \mid X=x}\\
        = &\expec{\expec{d(\omega, Y)^2 \mid X \in L(x)} - \expec{d(\omega, Y)^2 \mid X=x}},
    \end{align*}
    where $L(x) = L(x; \xi, D)$. Using Assumption \ref{ass-lip}, $M$ is Lipschitz continuous with constant $C(\omega)$ and we get
    \begin{equation*}
        \abs{\expec{d(\omega, Y)^2 \mid X \in L(x)} - \expec{d(\omega, Y)^2 \mid X=x}}
        \leq C(\omega) \t{diam}\,L(x).
    \end{equation*}
    Using now that the trees are constructed following Specification 1 of \citet{athey_generalized_2019}, Lemma 2 of \cite{wager_estimation_2018} applies, implying that the size of the leaves are $o_P(1)$. Hence
    \begin{equation*}
        \abs{\expec{T - \expec{d(\omega, Y)^2 \mid X=x}}}
        \leq \expec{C(\omega) L(x)} = \expec{o_P(1)} = o(1)
    \end{equation*}

    Altogether, we thus have that $\expec{\left(M_n^\star - M\right)^2} = o(1)$ and hence we have for any fixed $(\omega, x)$ the converge $M_n^\star(\omega; x) - M(\omega; x) = o_P(1)$.
\end{proof}

We can now prove that the random forest estimator is a pointwise constistent estimator of the Fr\'echet mean.

\begin{proof}[Proof of Theorem \ref{thm-consistency}]
    Since $\Omega$ is compact and $m_n^\star(x)$ minimizes $M_n^\star(\cdot; x)$, we have by definition that $M_n^\star(m_n^\star(x); x) \leq M_n^\star(m(x); x)$.
    Furthermore, by Theorem \ref{thm-unif}, the random forest objective function $M_n^\star$ converges uniformly to $M$. Together with the assumption that the Fr\'echet mean is well-separated, we have by Theorem \ref{thm-vdv} that $m_n^\star(x)$ is a consistent estimator of $m(x)$.
\end{proof}

\subsection{Consistency of the splitting rule}

We finally prove consistency of our medoid-based splitting rule. The proof follows a similar structure as the proof of Theorem \ref{thm-consistency}.

\begin{proof}[Proof of Proposition~\ref{thm-consistency-isfm}]
Let $\Omega_n = \eset{Y_1, \ldots, Y_n}$ and $\varepsilon > 0$. By assumption, there exists an $\alpha > 0$ such that $\Prob{d(Y, \omega_\oplus) < \varepsilon} > \alpha$. We then have by independence that
\begin{equation*}
    \Prob{d(\Omega_n, \omega_\oplus) > \varepsilon} = \Prob{d(Y, \omega_\oplus) > \varepsilon}^n < (1 - \alpha)^n \rightarrow 0,
\end{equation*}
showing that $d(\Omega_n, \omega_\oplus) \rightarrow 0$ in probability.

Let $M(\omega) = \expec{d(\omega, Y)^2}$ be the unweighted population Fréchet function, minimized by $\omega_\oplus$, and $M_n(\omega) = \frac{1}{n}\sum_{i=1}^n d(\omega, Y_i)^2$ be the unweighted empirical Fréchet functions with a sequence of minimizers $\eset{\hat\omega_n}$. We start by showing that $M$ is continuous, which follows directly from the Lipschitz continuity of the squared distance since for $\omega, \omega' \in \Omega$,
\begin{align*}
    \abs{M(\omega) - M(\omega')} 
    &= \abs{\expec{d(\omega, Y)^2} - \expec{d(\omega', Y)^2}} \\
    &\leq \expec{\abs{d(\omega, Y)^2 - d(\omega', Y)^2}} \\
    &\leq 2 \diamO\,d(\omega, \omega').
\end{align*}

Furthermore, we can also use the Lipschitz continuity of the squared distance to show that $M_n$ is stochastic equicontinuous since for each $\omega, \omega' \in \Omega$, 
\begin{align*}
    \abs{M_n(\omega) - M_n(\omega')}
    &= \abs{\sum_{i=1}^n d(\omega, Y_i)^2 - \sum_{i=1}^n d(\omega', Y_i)^2}\\
    &\leq \sum_{i=1}^n \abs{d(\omega, Y_i)^2 - d(\omega', Y_i)^2}\\
    &\leq \sum_{i=1}^n 2 \diamO\,d(\omega, \omega')\\
    &= 2 \diamO\,d(\omega, \omega').
\end{align*}
Finally, by the weak law of large numbers, $M_n$ converges pointwise to $M$, which verifies the conditions of Theorem \ref{thm-newey}. This shows that $\norm{M_n - M}_\Omega = o_P(1)$. Together with the assumption of well-separatedness of $\omega_\oplus$, this shows that $M_n$ fulfills the conditions of Theorem \ref{thm-vdv}, giving that $d(\hat\omega_n, \omega_\oplus) = o_P(1)$.

Let $\varepsilon > 0$ and $K = 2 \diamO$, under the event $E_n = \eset{ d(\Omega_n, \omega_\oplus) \leq \varepsilon/2K } \cap \eset{ d(\hat\omega_n, \omega_\oplus) \leq \varepsilon/2K }$, there exists for each $n \in \N$ an $\omega^\dagger_n \in \Omega_n$ such that $d(\omega^\dagger_n, \omega_\oplus) \leq \varepsilon / 2K$. Hence,
\begin{align*}
    \min_{\omega \in \Omega_n} M_n(\omega) - \min_{\omega \in \Omega} M_n(\omega)
    &= M_n(\tilde\omega_n) - M_n(\hat\omega_n)
    \leq M_n(\omega^\dagger_n) - M_n(\hat\omega_n)\\
    &\leq K d(\omega^\dagger_n, \hat\omega_n)
    \leq K (d(\omega^\dagger_n, \omega_\oplus) + d(\omega_\oplus, \hat\omega_n))
    \leq \varepsilon.
\end{align*}
By the two first results in this proof, $\Prob{E_n} \rightarrow 1$ and we obtain that $M_n(\tilde\omega_n) - M_n(\hat\omega_n) = o_P(1)$, and therefore
\begin{equation*}
     M_n(\tilde\omega_n) = M_n(\hat\omega_n) + o_P(1) \leq M_n(\omega_\oplus) + o_P(1).
\end{equation*}
This shows that the sequence of minimizers $\eset{\tilde\omega_n}$ fulfills the third condition of Theorem \ref{thm-vdv}, hence $d(\tilde\omega_n, \omega_\oplus) = o_P(1)$.
\end{proof}

\bibliographystyle{elsarticle-harv}
\bibliography{main}

\end{document}